\documentclass[a4paper,UKenglish,cleveref, autoref, thm-restate]{lipics-v2021}

\title{Quadratic Kernel for Cliques or Trees Vertex Deletion} 

\author{Soh Kumabe}{CyberAgent, Japan}{kumabe_soh@cyberagent.co.jp}{https://orcid.org/0000-0002-1021-8922}{}

\authorrunning{S. Kumabe} 

\Copyright{Soh Kumabe} 



\category{} 

\relatedversion{} 

\nolinenumbers 

\EventEditors{John Q. Open and Joan R. Access}
\EventNoEds{2}
\EventLongTitle{42nd Conference on Very Important Topics (CVIT 2016)}
\EventShortTitle{CVIT 2016}
\EventAcronym{CVIT}
\EventYear{2016}
\EventDate{December 24--27, 2016}
\EventLocation{Little Whinging, United Kingdom}
\EventLogo{}
\SeriesVolume{42}
\ArticleNo{23}

\usepackage{graphicx} 
\usepackage{amsmath,amssymb,amsthm,amsfonts}
\usepackage{bm,hyperref}
\usepackage{color}
\usepackage{booktabs,multirow}
\usepackage{comment}
\usepackage[appendix=inline]{apxproof}

\usepackage{algpseudocode}
\usepackage[linesnumbered,ruled,vlined]{algorithm2e}
\SetKwProg{Procedure}{Procedure}{}{}

\allowdisplaybreaks

\newcommand{\DP}{\mathsf{DP}}
\newcommand{\AUX}{\mathsf{AUX}}

\newtheorem{rrule}{Reduction Rule}
\newtheoremrep{reptheorem}[theorem]{Theorem}
\newtheoremrep{replemma}[theorem]{Lemma}

\begin{document}

\maketitle

\begin{abstract}
We consider \textsc{Cliques or Trees Vertex Deletion}, which is a hybrid of two fundamental parameterized problems: \textsc{Cluster Vertex Deletion} and \textsc{Feedback Vertex Set}.  
In this problem, we are given an undirected graph $G$ and an integer $k$, and asked to find a vertex subset $X$ of size at most $k$ such that each connected component of $G-X$ is either a clique or a tree.  
Jacob et al. (ISAAC, 2024) provided a kernel of $O(k^5)$ vertices for this problem, which was recently improved to $O(k^4)$ by Tsur (IPL, 2025).  

Our main result is a kernel of $O(k^2)$ vertices.  
This result closes the gap between the kernelization result for \textsc{Feedback Vertex Set}, which corresponds to the case where each connected component of $G-X$ must be a tree.  

Although both \emph{cluster vertex deletion number} and \emph{feedback vertex set number} are well-studied structural parameters, little attention has been given to parameters that generalize both of them.  
In fact, the lowest common well-known generalization of them is clique-width, which is a highly general parameter.  
To fill the gap here, we initiate the study of the \emph{cliques or trees vertex deletion number} as a structural parameter.  
We prove that \textsc{Longest Cycle}, which is a fundamental problem that does not admit $o(n^k)$-time algorithm unless ETH fails when $k$ is the clique-width, becomes fixed-parameter tractable when parameterized by the cliques or trees vertex deletion number.

\end{abstract}

\newpage

\section{Introduction}

Given a graph, can we remove at most $k$ vertices so that the resulting graph belongs to a class $\Pi$ of well-structured graphs?
Such problems are called \emph{vertex deletion problems} and include many well-studied parameterized problems.
Indeed, \textsc{Vertex Cover}~\cite{fellows2018known,Harris24}, \textsc{Feedback Vertex Set}~\cite{downey1995fixed,iwata17,li2022detecting,thomasse20104}, \textsc{Cluster Vertex Deletion}~\cite{bessy2023kernelization,fomin2019subquadratic,huffner2010fixed}, \textsc{Odd Cycle Transversal}~\cite{kratsch2012representative,reed2004finding}, \textsc{Interval Vertex Deletion}~\cite{cao2015interval}, and \textsc{Chordal Vertex Deletion}~\cite{marx2010chordal} correspond to the cases where $\Pi$ is the classes of collection of isolated vertices, collection of trees, collection of cliques, bipartite graphs, interval graphs, and chordal graphs, respectively.

However, why must $\Pi$ be a single graph class?
It is still reasonable to call a graph well-structured when all connected components are well-structured, even if different components belong to different classes.  
Jacob et al.~\cite{jacob2023deletion1} introduced the problem framework of \emph{deletion to scattered graph classes} capturing this concept, asking: 
Can we remove at most $k$ vertices from a given graph so that each connected component of the resulting graph belongs to one of the graph classes $(\Pi_1, \dots, \Pi_p)$?
Together with the subsequent paper~\cite{jacob2023deletion2}, they investigated the parameterized complexity of this type of problems and provided both general and problem-specific algorithms.  


This paper considers the following fundamental special case of deletion problems to scattered graph classes, called \textsc{Cliques or Trees Vertex Deletion}, which is first studied in~\cite{jacob2023deletion2}.

\begin{quote}
\textsc{Cliques or Trees Vertex Deletion}:  
Given an undirected graph $G=(V,E)$ and an integer $k\in \mathbb{Z}_{\geq 1}$, is there a vertex subset $X\subseteq V$ with $|X|\leq k$ such that each connected component of $G-X$ is either a clique or a tree?
\end{quote}

This case is particularly interesting because it combines two of the most prominent parameterized problems, \textsc{Feedback Vertex Set} and \textsc{Cluster Vertex Deletion}.
Moreover, since both the \emph{feedback vertex set number} and the \emph{cluster vertex deletion number} are well-studied structural parameters, it is natural to expect that the \emph{cliques or trees vertex deletion number}, which is the minimum $k$ such that $(G,k)$ becomes an $\mathsf{yes}$-instance of \textsc{Cliques or Trees Vertex Deletion}, is likewise an interesting structural parameter.
In particular, this number captures the structural simplicity of graphs that contain both dense parts (i.e., clique components) and sparse parts (i.e., tree components).
This illustrates the benefit of considering deletion to scattered graph classes, as a deletion distance to a single dense or sparse class alone cannot capture such a property.

It has been proved that this problem is in FPT~\cite{jacob2023deletion1,jacob2023deletion2}; that is, it admits an $f(k)n^c$-time algorithm for some computable function $f$ and constant $c$. 
Accordingly, researchers are investigating kernelization algorithms, which are polynomial-time algorithms that reduce the given instance to an equivalent instance of size $g(k)$ for some function $g$.
Jacob et al.~\cite{jacob2024kernel} showed that this problem admits a kernel with $O(k^5)$ vertices, which was later improved to $O(k^4)$ vertices by Tsur~\cite{tsur2025faster}.
However, there is still a significant gap when compared to the kernelization results of the original two problems, \textsc{Feedback Vertex Set} and \textsc{Cluster Vertex Deletion}, which admit a kernel with $O(k^2)$ vertices~\cite{iwata17,thomasse20104} and $O(k)$ vertices~\cite{bessy2023kernelization}, respectively.

In this paper, we close this gap by proving the following.

\begin{theorem}\label{thm:main}
\textsc{Cliques or Trees Vertex Deletion} admits a kernel with $O(k^2)$ vertices.
\end{theorem}

We remark that our result may surpass the expectations in the original research by Jacob et al.~\cite{jacob2024kernel}, as they stated the following in their conclusion section:
\begin{quote}
     One natural open question is to improve the size of our kernel, e.g. to $O(k^3)$ vertices. We believe that such a result is possible to achieve, but we suspect that it would require new techniques to develop such results. 
\end{quote}

We further initiate the study of \textsc{Cliques or Trees Vertex Deletion} as a structural parameter.
As mentioned above, both the feedback vertex set number~\cite{GalbyKIST23,Zehavi23} and the cluster vertex deletion number~\cite{doucha2012cluster,Gima0KMV23} are well-studied structural parameters. However, their common generalization has received relatively little attention.
Actually, the smallest well-investigated class that includes both graphs with bounded feedback vertex set number and graphs with bounded cluster vertex deletion number is the class of bounded \emph{clique-width}, which is a very general parameter often placed at the top of diagrams illustrating the inclusion relationships among structural parameters of graphs (see Figure~\ref{fig:diagram}).

As the second contribution of this paper, we demonstrate that the cliques or trees vertex deletion number is indeed a useful structural parameter. 
Particularly, we demonstrate that when parameterized by the cliques or trees vertex deletion number, a fundamental problem that is hard when parameterized by clique-width becomes tractable.
Specifically, we consider the following \textsc{Longest Cycle}, one of the most well-investigated problems in the field of parameterized complexity under structural parameterizations~\cite{bergougnoux2020optimal,cygan2018fast,cygan2022solving,doucha2012cluster,fomin2014almost,fomin2018clique,ganian2015improving,golovach2020graph,Jacob0Z23,lampis2011algorithmic}.  
\begin{quote}{\textsc{Longest Cycle}:}
Given an undirected graph $G=(V,E)$, find a cycle of $G$ with the largest possible number of vertices.
\end{quote}
It is known that, when $k$ is the clique-width, \textsc{Longest Cycle} (and even the special case \textsc{Hamiltonian Cycle}) does not admit an $n^{o(k)}$-time algorithm unless ETH fails~\cite{fomin2014almost,fomin2018clique}.
We prove the following.

\begin{theorem}\label{thm:cycle}
\textsc{Longest Cycle} can be solved in $2^{O(k\log k)}$-time, where $k$ is the cliques or trees vertex deletion number.
\end{theorem}

Note that this result generalizes, without losing time complexity, the FPT algorithm for \textsc{Hamiltonian Cycle} parameterized by the cluster vertex deletion number given by Doucha and Kratochv\'{i}l~\cite{doucha2012cluster} in the following two aspects.
First, we solve \textsc{Longest Cycle}, which is a generalization of \textsc{Hamiltonian Cycle}.
Second, our algorithm is parameterized by cliques or trees vertex deletion number, which is a more general parameter.
We further note that it is not hard to see that the clique-width of graphs with bounded cliques or trees vertex deletion number is bounded.

\begin{figure}[t]
  \centering
  \includegraphics[width=10cm]{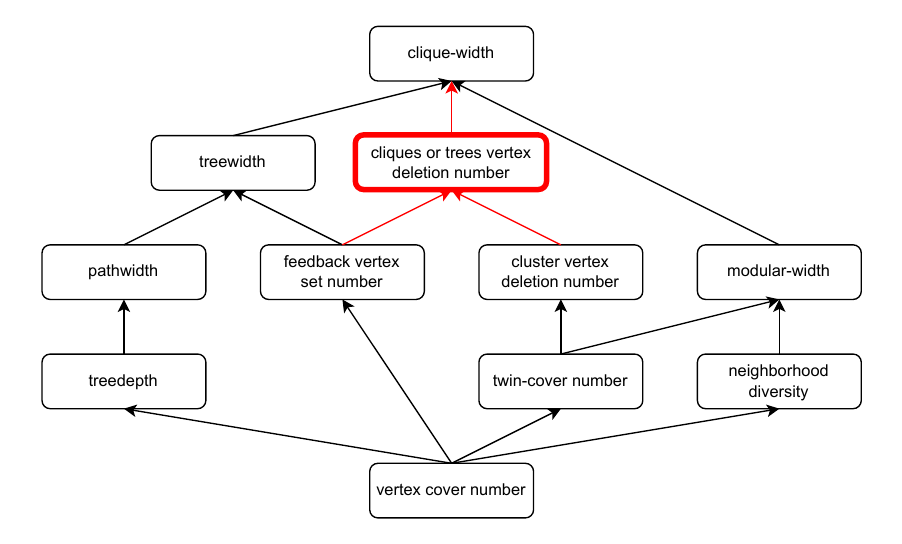}
  \caption{Diagram of structural parameters including the cliques or trees vertex deletion number.}
  \label{fig:diagram}
\end{figure}

\subsection{Related Work}
The literature on $\Pi$-deletion problems has studied several variants of \textsc{Feedback Vertex Set} and \textsc{Cluster Vertex Deletion}.  
Examples of extensions of \textsc{Feedback Vertex Set} are the cases where $\Pi$ consists of graphs with treewidth at most $\eta$~\cite{fomin2012planar,kim2015linear} and graphs that are at most $l$ edges away from forests~\cite{RaiS18}.  
Examples of variants of \textsc{Cluster Vertex Deletion} are the cases where $\Pi$ consists of $s$-plexes~\cite{liu2012editing}, $s$-clubs~\cite{guo2010more}, and graphs with a small dominating set number~\cite{bentert2024breaking}.  

Jacob et al.~\cite{jacob2023deletion1} introduced the notion of \emph{deletion to scattered graph classes}.  
In that paper, they obtained two general tractability results, stating that $(\Pi_1,\dots, \Pi_p)$\textsc{-Vertex Deletion} (i) is fixed-parameter tractable when each $\Pi_i$\textsc{-Vertex Deletion} is fixed-parameter tractable, and (ii) admits a $2^{poly(k)}$-time algorithm when each $\Pi_i$ is defined by a finite set of forbidden subgraphs.  
Jansen et al.~\cite{JansenKW25} improved the time complexity of result (ii) to $2^{O(k)}$.  
In the subsequent work~\cite{jacob2023deletion2}, Jacob et al. considered several specific cases and obtained efficient FPT algorithms and approximation algorithms.  
In particular, for \textsc{Cliques or Trees Vertex Deletion}, they obtained an $O^*(4^k)$-time FPT algorithm and a polynomial-time $4$-approximation algorithm.  
Jacob et al.~\cite{jacob2024kernel} provided a kernel with $O(k^5)$ vertices for this problem.  
Tsur~\cite{tsur2025faster} improved both of FPT and kernelization results by presenting a deterministic algorithm running in $O(3.46^k)$ time, a randomized algorithm running in $O(3.103^k)$ time, and a kernel with $O(k^4)$ vertices.  

Several studies have investigated the parameterized complexity of \textsc{Longest Cycle} and its special case, \textsc{Hamiltonian Cycle}, under structural parameterizations.  
Parameterized by the cluster vertex deletion number, Doucha and Kratochv\'{i}l~\cite{doucha2012cluster} presented an FPT algorithm for \textsc{Hamiltonian Cycle}. 
Parameterized by the clique-width, Fomin et al.~\cite{fomin2014almost,fomin2018clique} showed that no $n^{o(k)}$-time algorithm for \textsc{Hamiltonian Cycle} exists unless ETH fails. Bergougnoux et al.~\cite{bergougnoux2020optimal} proved this bound is tight by giving an $n^{O(k)}$-time algorithm.
The cases parameterized by treewidth~\cite{cygan2022solving}, pathwidth~\cite{cygan2018fast}, twin-cover number~\cite{ganian2015improving}, and proper interval deletion number~\cite{golovach2020graph} are also investigated, as well as directed tree-width~\cite{lampis2011algorithmic} and directed feedback vertex set number~\cite{Jacob0Z23} for the directed version.  

\subsection{Technical Overview}
Here, we provide a technical overview of Theorems~\ref{thm:main} and~\ref{thm:cycle}.  
Section~\ref{sec:overview_kernel} also contains a brief explanation of how our kernelization algorithm differs from the existing kernelization algorithms by Jacob et al.~\cite{jacob2024kernel} and Tsur~\cite{tsur2025faster}.  

\subsubsection{Overview of Theorem~\ref{thm:main}}\label{sec:overview_kernel}

Here, we briefly explain our idea toward Theorem~\ref{thm:main}.  
Our approach, at the highest level, is based on the following simple observation.  
Let $v$ be a vertex, and let $N_G(v):=\{u\in V\colon (u,v)\in E\}$ be the set of neighbors of $v$.  
Assume there is a feasible solution\footnote{We say a solution $X$ is \emph{feasible} if $|X|\leq k$ and each connected component of $G-X$ is either a clique or a tree.} $X$ such that $v$ is in a connected component of $G-X$ that is a clique.  
Then, the neighbors of $v$ in $G-X$ induce a clique.  
In particular, $N_G(v)$ contains a clique of size at least $|N_G(v)\setminus X|\geq |N_G(v)|-k$.  
Similarly, assume there is a feasible solution $X$ such that $v$ is in a connected component of $G-X$ that is a tree.  
Then, the neighbors of $v$ in $G-X$ induce an independent set.  
In particular, $N_G(v)$ contains an independent set of size at least $|N_G(v)\setminus X|\geq |N_G(v)|-k$.  

The core observation is that if $|N_G(v)|$ is large, say at least $2k+2$, these two situations cannot occur simultaneously, as a clique and an independent set cannot share an edge.  
In other words, if the degree of $v$ is large, we can determine that $v$ is either  
\begin{itemize}
    \item always in a clique component of $G-X$ unless $v\in X$, or  
    \item always in a tree component of $G-X$ unless $v\in X$,  
\end{itemize}
for all feasible solutions $X$.  

In our algorithm, we first partition the vertices into three subsets. Let $c=7$. Set $V_{\mathrm{ld}}$ of vertices with a high degree ($ > ck$) and dense neighbors, set $V_{\mathrm{ls}}$ of vertices with a high degree ($ > ck$) and sparse neighbors,  
and set $V_{\mathrm{small}}$ of vertices with a low degree ($\leq ck$).  
Using the above observation, we can claim that, for any feasible solution $X$, vertices in $V_{\mathrm{ld}}$ cannot be in a tree component of $G-X$, and vertices in $V_{\mathrm{ls}}$ cannot be in a clique component of $G$ of $G-X$.  
This observation, very roughly speaking, enables us to apply reduction rules to $V_{\mathrm{ld}}$ as if we were solving \textsc{Cluster Vertex Deletion}.
Similarly, we can apply reduction rules to $V_{\mathrm{ls}}$ as if we were solving \textsc{Feedback Vertex Set}.  
Furthermore, we can reduce the number of vertices in $V_{\mathrm{small}}$ using the fact that they have small degrees.  

Here we go into a little more detail.
Our kernelization algorithm proceeds as follows.  
First, we reduce the degrees of vertices in $V_{\mathrm{ls}}$ into $O(k)$.
To acheive this, we directly apply reduction rules used in the celebrated quadratic kernelization of \textsc{Feedback Vertex Set} by Thomass\'{e}~\cite{thomasse20104}.
Note that this direct applicability is already a benefit of partitioning the vertex set.
Indeed, the original kernelization algorithm by Jacob et al.~\cite{jacob2024kernel} also relied on Thomass\'{e}’s quadratic kernelization in its ``tree part'' (and Tsur's improvement~\cite{tsur2025faster} did not touch this part), but extended the algorithm and analysis, including the use of a new version of the expansion lemma by Fomin et al.~\cite{fomin2019subquadratic}.
In contrast, ours is almost identical to Thomass\'{e}’s, and its analysis is also nearly the same except for a few technical details, making our approach significantly simpler.  

Next, we reduce the number of vertices in $V_{\mathrm{ld}}$.
The main technical contribution of this paper lies in this part, since the kernelization algorithms by Jacob et al.~\cite{jacob2024kernel} and Tsur~\cite{tsur2025faster} bound the number of vertices in the ``clique part'' by $O(k^5)$ and $O(k^4)$, respectively, which dominates the overall kernel size.  
Both of those algorithms are based on a marking procedure, whereas ours is not.  
As in those algorithms, we first compute a constant factor approximate solution $S$ in polynomial time and reduce the number of clique components of $G-S$ by $O(k)$ using the textbook reduction~\cite{fomin2019kernelization} for \textsc{Cluster Vertex Deletion} using the expansion lemma.
Now, we can argue that there are only $O(k^2)$ vertices in $V_{\mathrm{ld}}$ that are adjacent to some vertex outside $V_{\mathrm{ld}}$ as follows. 
First, any such vertex must
\begin{itemize}
    \item[(i)] belongs to $S$,  
    \item[(ii)] be adjacent to a vertex in $S\setminus V_{\mathrm{ld}}$, or
    \item[(iii)] belongs to a clique component of $G-S$ that contains a vertex outside $V_{\mathrm{ld}}$.
\end{itemize}
The number of vertices satisfying (i) is $O(k)$ by definition.  
Since every vertex outside $V_{\mathrm{ld}}$ has degree $O(k)$, the number of vertices satisfying (ii) is bounded by $O(k^2)$.  
Furthermore, each clique component that appears in (iii) contains at most $O(k)$ vertices because it contains a vertex with degree $O(k)$. Together with the fact that the number of clique components is $O(k)$, the number of vertices satisfying (iii) is bounded by $O(k^2)$.
Now, we concentrate on reducing the number of vertices in $V_{\mathrm{ld}}$ that are adjacent only to vertices in $V_{\mathrm{ld}}$ by $O(k^2)$.
To do this, we first borrow ideas from the quadratic kernelization of the \textsc{3-Hitting Set} by Abu-Khzam~\cite{abu2010kernelization} and construct a list $\mathcal{P}$ of induced $P_3$s (that are, induced subgraphs that are paths of length two) such that no two different $P_3$ share more than one vertex.
A standard argument similar to that used in Buss Kernel~\cite{buss1993nondeterminism} for \textsc{Vertex Cover} reduces the size of $\mathcal{P}$ to $O(k^2)$.
Let $V_{\mathrm{ldmod}}:=V_{\mathrm{ld}}\setminus \bigcup_{P\in \mathcal{P}}P$.
By introducing additional structural observations, we can state that 
\begin{itemize}
    \item there are at most $O(k)$ connected components in $V_{\mathrm{ldmod}}$, and
    \item each connected component $C$ of $V_{\mathrm{ldmod}}$ is a \emph{clique-module} of $V_{\mathrm{ld}}$, that is, the set $N_{V_{\mathrm{ld}}}(v)\cup \{v\}$ is same for all $v\in C$ and includes $C$.
\end{itemize}
The important observation is that if $G$ contains a clique-module of $V$ of size at least $k+4$, then one of its vertices can be safely removed.  
This rule seems to bound the size of each clique-module by $O(k)$ and bound the size of $V_{\mathrm{ldmod}}$ by $O(k^2)$. 
However, it is still insufficient because we want clique-modules of $V$ for the reduction, whereas the observation above only provides clique-modules of $V_{\mathrm{ld}}$.
Nevertheless, if a clique-module in $V_{\mathrm{ld}}$ consists only of vertices with no neighbor outside $V_{\mathrm{ld}}$, then it is also a clique-module of $V$. 
Therefore, we can reduce the number of vertices in each clique-module that have no neighbor outside $V_{\mathrm{ld}}$ to $O(k)$, obtaining the desired bound. 
(Precisely speaking, we need to perform a slightly more careful argument to deal with multi-edges that may be introduced in reductions for $V_{\mathrm{ls}}$.)

We reduce the number of vertices in $V\setminus V_{\mathrm{ld}}$ to $O(k^2)$ to complete the kernelization. 
By an argument similar to that for (iii) above, we can bound the number of vertices in $V\setminus V_{\mathrm{ld}}$ that belong to clique components of $G-S$ by $O(k^2)$.  
Therefore, it remains to bound the number of vertices that belong to tree components.
To do this, we apply a few reduction rules based on local structure. Most of these rules appear in~\cite{jacob2024kernel}, but one is new. 
We apply the following standard argument used in the kernelization of \textsc{Feedback Vertex Set}:  
If a graph with minimum degree at least $3$ can be turned into a forest by removing $k$ vertices of degree at most $t$, then the number of vertices in the original graph was $O(tk)$.  
However, in this case, we cannot completely eliminate vertices of degree $1$ or $2$, so this argument cannot be applied directly.  
Nevertheless, we can bound the number of such low-degree vertices using the numbers of vertices of degree $3$ and at least $4$.  
By extending the argument for \textsc{Feedback Vertex Set} using these fine-grained bounds, we can bound the number of vertices in $V\setminus V_{\mathrm{ld}}$ by $O(k^2)$, which completes the analysis.  

\subsection{Overview of Theorem~\ref{thm:cycle}}
Here, we briefly explain our ideas toward Theorem~\ref{thm:cycle}.  
Let $G=(V,E)$ be a graph and $X\subseteq V$ be a given vertex subset with $|X| = k$ such that each connected component of $G-X$ is either a clique or a tree.  
We begin by brute-force the order in which the desired cycle $P$ visits the vertices in $X$.  
If $P$ is disjoint from $X$, the problem is trivial, so we may assume that $P$ intersects $X$, and denote the vertices in $X\cap P$ by $v_1,\dots, v_l$ in the order they appear along $P$.  
Let $P_i$ be the $v_i,v_{i+1}$-path appearing on $P$ (indices modulo $l$).  
Then, unless $P_i$ has length $1$, the internal vertices of $P_i$ belong to a single connected component $C_i$ of $G-X$.  

The next step is limiting the candidates of $C_i$. Here, we can state that, as candidates of $C_i$, it is sufficient to consider the top $k$ components that admit the longest $v_i,v_{i+1}$-paths.
Then, we can brute-force over all tuples $(C_1,\dots,C_l)$ within a total cost of $2^{O(k\log k)}$.  
Now, the problem is reduced to solving the following problem for each connected component $H$ of $G-X$, where we set $J_H := \{i \colon C_i = H\}$.
\begin{quote}
Given a list of vertex subset pairs $\{(V_{i1} := N(v_i)\cap C_i,\ V_{i2} := N(v_{i+1})\cap C_i)\}_{i\in J_H}$, compute the maximum total length of vertex-disjoint paths $\{P_i\}_{i\in J_H}$ such that each $P_i$ is a path from a vertex in $V_{i1}$ to a vertex in $V_{i2}$.
\end{quote}

We give FPT algorithms for this problem on both cliques and trees.
First, we explain the algorithm for cliques.
We brute-force over flags $f \in \{0,1\}^{J_H}$, where $f_i = 0$ represents that $P_i$ consists of a single vertex, and $f_i = 1$ indicates that $P_i$ contains at least two vertices.  
The problem of whether a family of paths satisfying the conditions defined by $f$ exists can be reduced to the bipartite matching problem and, thus, solved in polynomial time.  
If such a family of paths exists for some flag $f$ such that $f_i = 1$ holds for some $i \in J_H$, we can extend the family to cover all vertices in $H$ by appropriately adding internal vertices to the paths. Thus, in this case, the answer is $|H|-|J_H|$.
If such a family of paths exists only for $f = (0, \dots, 0)$, the answer is zero.  
If no such family of paths exists for any $f$, then the problem is infeasible. 
The time complexity is $O^*(2^k)$.

Now, we explain the algorithm for trees.
Our algorithm uses dynamic programming.
We omit the formal details here, but the intuition is as follows.
Regard $H$ as a rooted tree.  
For a vertex $v$ and a set $Z \subseteq J_H$, we define $\DP[v][Z]$ denote the maximum total length of a family of paths $\{P_i\}_{i \in Z}$ that can be packed into the subtree rooted at $v$.
We compute these values in a bottom-up manner.
We can analyze that such dynamic programming can be implemented to work in $O^*(3^k)$ time. 

\subsection{Organization}

The rest of this paper is organized as follows.  
In Section~\ref{sec:prelim}, we introduce basic notation and well-known techniques in the literature on kernelization.  
In Section~\ref{sec:kernelization}, we prove Theorem~\ref{thm:main} by constructing a kernel for \textsc{Cliques or Trees Vertex Deletion} with $O(k^2)$ vertices. 
In Section~\ref{sec:cycle}, we prove Theorem~\ref{thm:cycle} by presenting an FPT algorithm for \textsc{Longest Cycle} parameterized by the cliques or trees vertex deletion number.  
\section{Preliminaries}\label{sec:prelim}

In this paper, the term \emph{graph} refers to an undirected graph, which does not contain self-loops but may contain multi-edges.
For a graph $G=(V,E)$ and a vertex $v\in V$, we call a vertex belonging to $N_G(v):=\{u\in V\colon (u,v)\in E\}$ a \emph{neighbor} of $v$ in $G$, and the number of edges incident to $v$, denoted $d_G(v)$, is referred to as the \emph{degree} of $v$.
Note that $|N_G(v)|$ and $d_G(v)$ may differ due to multi-edges.
Moreover, we denote by $\rho_G(v)$ the size of the set $\{\{u_1,u_2\}\subseteq N_G(v)\colon (u_1,u_2)\in E\}$. 
In other words, $\rho_G(v)$ represents the number of edges connecting the neighbors of $v$ in $G$, where multi-edges are counted as a single edge.
When the context is clear, we omit the subscript $G$ and write $N(v)$, $d(v)$, and $\rho(v)$ for simplicity.
A vertex $v$ with $d(v)=1$ is called a \emph{pendant}.

For a vertex subset $Z\subseteq V$, we denote $E(Z):=\{e\in E\colon e\subseteq Z\}$. 
The \emph{subgraph of $G$ induced by $Z$} is the graph $(Z,E(Z))$.
For simplicity, when there is no risk of confusion, we identify the vertex set $Z$ with the subgraph induced by it: that is, when we refer to the ``graph $Z$'' for $Z\subseteq V$, we mean the subgraph induced by $Z$.
Moreover, for a vertex subset $Z\subseteq V$, we write $G-Z$ to denote the subgraph induced by $V\setminus Z$. 
When $Z$ consists of a single vertex $z$, we abbreviate $G-\{z\}$ as $G-z$.
Similarly, we denote by $G-e$ the graph obtained by removing an edge $e\in E$ from $G$.
A vertex subset $Z\subseteq V$ induces a \emph{clique} if there is exactly one edge between any two different vertices in $Z$, an \emph{independent set} if there is no edge between any two vertices in $Z$, and a \emph{tree} if $Z$ is connected and contains no cycle.  

For a vertex $v\in V$ and $t\in \mathbb{Z}_{\geq 1}$, a \emph{$v$-flower} of \emph{order} $t$ is a set of $t$ cycles passing through $v$ such that no two cycles share a common vertex other than $v$. 
The following is well-known.
\begin{lemma}[\rm{Gallai's Theorem~\cite{cygan2015parameterized,fomin2019kernelization,thomasse20104}}]\label{lem:gallai}
Given an undirected graph $G=(V,E)$, a vertex $v\in V$, and an integer $t\in \mathbb{Z}_{\geq 1}$, there is a polynomial-time algorithm that computes either
\begin{itemize}
    \item a $v$-flower of order $t+1$, or
    \item a vertex set $B\subseteq V\setminus \{v\}$ of size at most $2t$ such that $G-B$ contains no cycle passing through $v$.
\end{itemize}
\end{lemma}

For vertex subsets $K,L$ and an integer $q\in \mathbb{Z}_{\geq 1}$, an edge set $M$ is a \emph{$q$-expansion of $K$ into $L$} if 
\begin{itemize}
    \item exactly $q$ edges of $M$ are incident to each vertex in $K$, and
    \item exactly one edge of $M$ is incident to each vertex in $L$.
\end{itemize}
The following lemma is also well-known in the literature on kernelization.

\begin{lemma}[\rm{Expansion Lemma~\cite{cygan2015parameterized,fomin2019kernelization}}]
Let $H:=(K\dot{\cup} L, E)$ be a bipartite graph with vertex bipartition $(K,L)$ and $q\in \mathbb{Z}_{\geq 1}$. 
Assume $|L|\geq q|K|$ and $L$ contains no isolated vertex.
Then, there is a polynomial-time algorithm that computes a pair of non-empty vertex subsets $K'\subseteq K$ and $L'\subseteq L$ such that 
\begin{itemize}
    \item $N_H(L')\subseteq K'$, and
    \item there exists a $q$-expansion of $K'$ into $L'$.
\end{itemize}
\end{lemma}
\section{Quadratic Kernel for \textsc{Cliques or Trees Vertex Deletion}}\label{sec:kernelization}
Let $G=(V,E)$ be a graph and $k\in \mathbb{Z}_{\geq 1}$.
In this section, we prove Theorem~\ref{thm:main} by constructing a quadratic kernel for \textsc{Cliques or Trees Vertex Deletion}.

\subsection{Partitioning Vertices}\label{sec:partition}
We begin by classifying the vertices into three categories.
Let 
\begin{align*}
    V_{\mathrm{ls}}&:=\left\{v\in V\colon |N(v)| > 7k\land \rho(v)\leq \frac{|N(v)|(|N(v)|-1)}{4}\right\},\\
    V_{\mathrm{ld}}&:=\left\{v\in V\colon |N(v)| > 7k\land \rho(v) > \frac{|N(v)|(|N(v)|-1)}{4}\right\},\\
    V_{\mathrm{small}}&:=\left\{v\in V\colon |N(v)| \leq 7k\right\}.
\end{align*}
The vertices belonging to the first, second, and third groups are referred to as \emph{large-sparse}, \emph{large-dense}, and \emph{small}, respectively.
The following lemma states that, for any feasible solution $X$, large-sparse vertices are included in either $X$ or a tree component of $G-X$.

\begin{lemma}\label{lem:sparse_is_tree}
Let $v\in V_{\mathrm{ls}}$. Then, for any feasible solution $X$ with $v\not \in X$, $v$ is in a tree component of $G-X$.
\end{lemma}
\begin{appendixproof}
Assume that $v$ is in a clique component $C$ of $G-X$.
Then, at least $|N(v)|-|X|\geq |N(v)|-k$ neighbors of $v$ belong to $C$ and thus, $N(v)$ contains a clique of size $|N(v)|-k$. 
Therefore, we have 
\begin{align*}
    \rho(v)\geq \frac{(|N(v)|-k)(|N(v)|-k-1)}{2}.
\end{align*}
However, we have
\begin{align*}
    &\rho(v)-\frac{(|N(v)|-k)(|N(v)|-k-1)}{2}
    \leq \frac{|N(v)|(|N(v)|-1)-2(|N(v)|-k)(|N(v)|-k-1)}{4}\\
    &\quad = \frac{-|N(v)|^2+(4k+1)|N(v)|-2k(k+1)}{4}
    \leq \frac{-49k^2+(28k^2+7k)-(2k^2+2k)}{4}\\
    &\quad = \frac{-23k^2+5k}{4} < 0,
\end{align*}
where the first and second inequalities follow from the assumption that $v$ is a large-sparse vertex.
\end{appendixproof}

The following lemma states a result symmetric to Lemma~\ref{lem:sparse_is_tree}, that is, for any feasible solution $X$, large-dense vertices are included in either $X$ or a clique component of $G-X$.

\begin{lemma}\label{lem:dense_is_clique}
Let $v\in V_{\mathrm{ld}}$. Then, for any feasible solution $X$ with $v\not \in X$, $v$ is in a clique component of $G-X$.
\end{lemma}
\begin{appendixproof}
Assume that $v$ is in a tree component $C$ of $G-X$.
Then, at least $|N(v)|-|X|\geq |N(v)|-k$ neighbors of $v$ belong to $C$ and thus, $N(v)$ contains an independent set of size $|N(v)|-k$. 
Therefore, we have 
\begin{align*}
    \frac{|N(v)|(|N(v)|-1)}{2}-\rho(v)\geq \frac{(|N(v)|-k)(|N(v)|-k-1)}{2}.
\end{align*}
However, we have
\begin{align*}
    &\left(\frac{|N(v)|(|N(v)|-1)}{2}-\rho(v)\right)-\frac{(|N(v)|-k)(|N(v)|-k-1)}{2}\\
    &\quad \leq \frac{|N(v)|(|N(v)|-1)-2(|N(v)|-k)(|N(v)|-k-1)}{4}
    \leq \frac{-23k^2+5k}{4} < 0,
\end{align*}
where the first and second inequalities follow from the assumption that $v$ is a large-dense vertex.
\end{appendixproof}

Several times throughout this paper, we use the following type of alternative evidence for a vertex belonging to a clique component or a tree component.

\begin{lemma}\label{lem:large_ind_or_clique}
Let $v\in V$ and $X$ be a feasible solution with $v\not \in X$. If $N(v)$ contains a clique of size $k+2$, then $v$ is in a clique component of $G-X$.
Similarly, if $N(v)$ contains an independent set of size $k+2$, then $v$ is in a tree component of $G-X$.
\end{lemma}
\begin{appendixproof}
If $N(v)$ contains a clique of size $k+2$, then at least two of its vertices must survive in $G-X$.  
Therefore, there exists a pair of vertices in $N_{G-X}(v)$ such that there is an edge between them, implying that the component containing $v$ cannot be a tree component.  
Similarly, if $N(v)$ contains an independent set of size $k+2$, then at least two of its vertices must survive in $G-X$.  
Therefore, there exists a pair of vertices in $N_{G-X}(v)$ such that there is no edge between them, implying that the component containing $v$ cannot be a clique component.  
\end{appendixproof}

In the rest of this paper, we will use Lemmas~\ref{lem:sparse_is_tree},~\ref{lem:dense_is_clique},~and~\ref{lem:large_ind_or_clique} as basic tools without specifically mentioning them.

\subsection{Bounding Sizes of Neighbors of Vertices in $V_{\mathrm{ls}}$}\label{sec:large-sparse}
In this section, we reduce the size of $N(v)$ for large-sparse vertices $v\in V_{\mathrm{ls}}$.
Most of the reduction rules in this section are the same as the quadratic kernelization of \textsc{Feedback Vertex Set} by Thomass\'{e}~\cite{thomasse20104}, while some details in the analysis require additional care in proofs.
We begin with the following.
\begin{rrule}\label{rrule:gallai}
Let $v$ be any large-sparse vertex. Apply Lemma~\ref{lem:gallai} for $v$ and $t=k$. If a $v$-flower of order $k+1$ is found, remove $v$ and decrease $k$ by $1$.
\end{rrule}

\begin{lemma}\label{lem:gallai_is_safe}
Reduction Rule~\ref{rrule:gallai} is safe.
\end{lemma}
\begin{appendixproof}
We prove that $v$ is contained in all solutions $X$.
Assume the contrary.
From Lemma~\ref{lem:sparse_is_tree}, $v$ belongs to a tree component of $G-X$.
Thus, $X$ must hit all cycles of the $v$-flower, and when $X$ does not contain $v$, its size would be at least $k+1$, leading to a contradiction.
\end{appendixproof}

Let $v$ be a large-sparse vertex and assume Lemma~\ref{lem:gallai} finds a vertex set $B$ with size at most $2k$ that hits all cycles containing $v$.
Let $\mathcal{C}_{\mathrm{tree}}$ be the family of connected components of $G-v-B$ that are trees and adjacent to $v$.
Similarly, let $\mathcal{C}_{\mathrm{nontree}}$ be the family of connected components of $G-v-B$ that are not trees and adjacent to $v$.
Since $G-B$ contains no cycle containing $v$, for each $C\in \mathcal{C}_{\mathrm{tree}}\cup \mathcal{C}_{\mathrm{nontree}}$, we have $|N(v)\cap C|=1$. 
Particularly, $|N(v)|=|\mathcal{C}_{\mathrm{tree}}|+|\mathcal{C}_{\mathrm{nontree}}|+|N(v)\cap B|$.
We bound $|N(v)|$ by bounding these three terms. Obviously, $|N(v)\cap B|\leq |B|\leq 2k$.
To bound $|\mathcal{C}_{\mathrm{nontree}}|$, we use the following reduction rule.

\begin{rrule}\label{rrule:nontree}
If $|\mathcal{C}_{\mathrm{nontree}}|\geq k+1$, remove $v$ and decrease $k$ by $1$.
\end{rrule}
\begin{lemma}\label{lem:nontree_is_safe}
Reduction Rule~\ref{rrule:nontree} is safe.
\end{lemma}
\begin{appendixproof}
We prove that $v$ is contained in all solutions $X$.
Assume otherwise.
Then, there is a component $C\in \mathcal{C}_{\mathrm{nontree}}$ with $X\cap C = \emptyset$.
Therefore, $G-X$ contains a connected component containing $C\cup \{v\}$, which contains a cycle and contradicts to Lemma~\ref{lem:sparse_is_tree}.
\end{appendixproof}

Now we bound $|\mathcal{C}_{\mathrm{tree}}|$ by $4k$ using the expansion lemma.
We construct an auxiliary bipartite graph $H$. The vertex set of $H$ is $B\dot{\cup}\mathcal{C}_{\mathrm{tree}}$ with bipartition $(B, \mathcal{C}_{\mathrm{tree}})$.
We add an edge between $b\in B$ and $C\in \mathcal{C}_{\mathrm{tree}}$ if and only if $N_G(b)\cap C\neq \emptyset$.
We use the following reduction rule to ensure the part $\mathcal{C}_{\mathrm{tree}}$ does not contain isolated vertices.
\begin{rrule}\label{rrule:tree_pruning}
If there is a component $C\in \mathcal{C}_{\mathrm{tree}}$ that has no neighbor in $B$, remove all vertices of $C$.
\end{rrule}
\begin{lemma}\label{lem:tree_pruning_is_safe}
Reduction Rule~\ref{rrule:tree_pruning} is safe.
\end{lemma}
\begin{appendixproof}
Clearly, if the original instance is $\mathsf{yes}$-instance, the reduced instance is also $\mathsf{yes}$-instance.
We prove the opposite direction. Let $X'$ be a feasible solution for $G-C$.
It is sufficient to prove that $v$ cannot be in a clique component of $G-C-X'$.
Since $v$ is large-sparse in $G$, we have $|N_{G}(v)|\geq 7k$ and thus, $|\mathcal{C}_{\mathrm{tree}}|+|\mathcal{C}_{\mathrm{nontree}}|\geq 7k-|N_G(v)\cap B|\geq 5k\geq k+3$.
Particularly, $N_{G-C}(v)$ contains an independent set of size at least $k+2$, and thus, $v$ cannot be in a clique component of $G-C-X'$.
\end{appendixproof}

Assume $|\mathcal{C}_{\mathrm{tree}}|\geq 4k$.
We apply $2$-expansion lemma to $H$ and obtain vertex sets $\mathcal{C}'\subseteq \mathcal{C}_{\mathrm{tree}}$ and $B'\subseteq B$ such that there is a $2$-expansion of $B'$ into $\mathcal{C}'$.
We have $|B'|\leq k$ because otherwise we obtain a $v$-flower of order $k+1$.
We apply the following reduction.
\begin{rrule}\label{rrule:expansion_reduction_tree}
Remove each edge between $v$ and $\mathcal{C}'$.
Then, connect $v$ and each vertex in $B'$ by a double-edge.
\end{rrule}
\begin{lemma}\label{lem:expansion_tree_is_safe}
Assume $|\mathcal{C}_{\mathrm{tree}}|\geq 4k$.
Then, Reduction Rule~\ref{rrule:expansion_reduction_tree} is safe.
\end{lemma}
\begin{appendixproof}
Let $G'$ be the graph obtained from $G$ by applying Reduction Rule~\ref{rrule:expansion_reduction_tree}.
We first prove that if $G$ is an $\mathsf{yes}$-instance, $G'$ is also an $\mathsf{yes}$-instance.
Let $X$ be the feasible solution of $G$.
If $v\in X$, $X$ is also a feasible solution of $G'$ because $G-v=G'-v$.
Assume $v\not \in X$.
We prove that $X':=\left(X\setminus \bigcup_{C
\in \mathcal{C}'}C\right)\cup B'$ is a feasible solution of $G'$.

First, we bound the size of $X'$. Since there is a $2$-expansion of $B'$ into $\mathcal{C}'$, there is a $v$-flower of order $|B'|$ in $G$, which is contained in $B'\cup \{v\}\cup \bigcup_{C
\in \mathcal{C}'}C$. Since $v$ is contained in a tree component of $G-X$, $X$ hits all cycles of this $v$-flower and therefore we have $\left|X\cap \left(B'\cup \bigcup_{C\in \mathcal{C}'}C\right)\right|\geq |B'|$. Therefore, $|X'|=\left|X\setminus \left(B'\cup \bigcup_{C\in \mathcal{C}'}C\right)\right| + |B'|\leq |X|$.

Now, we prove that each connected component of $G'-X'$ is either a clique or a tree.
Since $N_G\left(\bigcup_{C\in \mathcal{C}'}C\right)\subseteq B'\cup \{v\}$, each $C\in \mathcal{C'}$ is a connected component of $G'-X'$, which is a tree. Since $G-B'-\bigcup_{C\in \mathcal{C}'}C$ is an induced subgraph of $G$, each other connected component of $G'-X'$ is an induced subgraph of some connected component of $G-X$, which is a clique or a tree.

Conversely, let $X'$ be a feasible solution of $G'$.
We prove $X'$ is also a feasible solution of $G$.
Since $v$ is connected to all vertices in $B'$ by double-edge, either $v\in X'$ or $B'\subseteq X'$ holds.
If $v\in X'$, $X'$ is also a feasible solution of $G$ because $G-v=G'-v$.
Assume $B'\subseteq X'$ and $v\not \in X'$.
Since $G'-B'$ is obtained from $G-B'$ by removing edges between $v$ and $\bigcup_{C\in \mathcal{C}'}C$, each connected component of $G-X'$ that does not contain $v$ is also a connected component of $G'-X'$, and thus, is a clique or a tree.
We finish the proof by proving that the connected component of $G-X'$ that contains $v$ is a tree.
Since $N_{G'}(v)$ contains an independent set of size $|\mathcal{C}_{\mathrm{nontree}}|+|\mathcal{C}_{\mathrm{tree}}|-|\mathcal{C}'|\geq 7k-|B|-2|B'|\geq 3k\geq k+2$, $v$ is in a tree component $C'$ of $G'-X'$.
The desired connected component is obtained from $C'$ by attaching the trees of $\mathcal{C'}$ at $v$, and thus, a tree.
\end{appendixproof}
Since all the above reduction rules reduce the number of pairs of vertices connected by at least one edge, the reduction rules in this section can be applied only a polynomial number of times.
Thus, we have the following.
\begin{lemma}\label{lem:degree_bound_trees}
Let $G$ be the graph obtained by exhaustively applying all the above reduction rules.
Then, for all vertex $v\in V\setminus V_{\mathrm{ld}}$, we have $|N(v)|\leq 7k$ (and thus, $V_{\mathrm{ls}}=\emptyset$).
\end{lemma}
\begin{appendixproof}
If $v\in V_{\mathrm{small}}$, we have $|N(v)|\leq 7k$.
If $v\in V_{\mathrm{ls}}$, we have $|N(v)|=|\mathcal{C}_{\mathrm{tree}}|+|\mathcal{C}_{\mathrm{nontree}}|+|N(v)\cap B|\leq 4k+k+2k=7k$.
\end{appendixproof}

\subsection{Bounding $|V_{\mathrm{ld}}|$}\label{sec:large-dense}
In this section, we reduce the number of vertices in $V_{\mathrm{ld}}$.
This part is the main technical contribution of this paper.
As in~\cite{jacob2024kernel}, we apply a $4$-approximation algorithm for \textsc{Cliques or Trees Vertex Deletion} given in~\cite{jacob2023deletion2} and obtain an approximate solution $S\subseteq V$.
We apply the following reduction rule to ensure $|S|\leq 4k$, which is clearly safe.
\begin{rrule}\label{rrule:boundS}
If $|S| > 4k$, return $\mathsf{no}$.
\end{rrule}

Let $\mathcal{C}_{\mathrm{clique}}$ be the family of connected components of $G-S$ that are cliques of size at least $3$.
Similarly, let $\mathcal{C}_{\mathrm{tree}}$ be the family of connected components of $G-S$ that are trees.
The following rule ensures that each vertex $v\in V_{\mathrm{ld}}$ is contained in a clique component of $G-S$ unless $v\in S$.
\begin{rrule}\label{rrule:correspond_components}
If there is a vertex $v\in V_{\mathrm{ld}}$ that is in a tree component of $G-S$, remove $v$ and decrease $k$ by $1$.
\end{rrule}
\begin{lemma}\label{lem:correspond_components_is_safe}
Reduction Rule~\ref{rrule:correspond_components} is safe.
\end{lemma}
\begin{proof}
We prove that $v$ is contained in all solutions $X$.
Since $v\in V_{\mathrm{ld}}$, $v$ is in a clique component of $G-X$ unless $v\in X$.
Since $v$ is in a tree component of $G-S$, $N_{G-S}(v)$ is an independent set, whose size is at least $|N_G(v)|-|S|\geq 7k-4k=3k\geq k+2$. Thus, $v$ is in a tree component of $G-X$ unless $v\in X$, which leads to $v\in X$.
\end{proof}

We construct an auxiliary bipartite graph $H$ as follows. The vertex set of $H$ is $S\dot{\cup}\mathcal{C}_{\mathrm{clique}}$ with bipartition $(S,\mathcal{C}_{\mathrm{clique}})$. We add an edge between $s\in S$ and $C\in \mathcal{C}_{\mathrm{clique}}$ if and only if $N_G(s)\cap C\neq \emptyset$.
To ensure the part $\mathcal{C}_{\mathrm{clique}}$ does not contain isolated vertices, we apply the following fundamental rule, which is clearly safe.
\begin{rrule}\label{rrule:trivial_component_removal}
If $G$ contains a connected component that is a clique or a tree, remove that component.
\end{rrule}

Assume $|\mathcal{C}_{\mathrm{clique}}|\geq 2|S|$.
We apply $2$-expansion lemma to $H$ and obtain vertex sets $\mathcal{C}'\subseteq \mathcal{C}_{\mathrm{clique}}$ and $S'\subseteq S$.
The following reduction is proved to be safe in~\cite{jacob2024kernel}.
\begin{rrule}\label{rrule:expansion_reduction_clique}
Remove vertices of $S'$ and decrease $k$ by $|S'|$. 
\end{rrule}
\begin{lemma}[\rm{\cite{jacob2024kernel}}]
Reduction Rule~\ref{rrule:expansion_reduction_clique} is safe.
\end{lemma}

Now we can assume $|\mathcal{C}_{\mathrm{clique}}|\leq 8k$.
We first bound the number of vertices in $V_{\mathrm{ld}}$ that are adjacent to some vertices outside $V_{\mathrm{ld}}$. We have the following.
\begin{lemma}\label{lem:large_dense_outside}
After applying all the above reduction rules, there are at most $84k^2+4k$ vertices in $V_{\mathrm{ld}}$ that have some neighbor outside $V_{\mathrm{ld}}$.
\end{lemma}
\begin{proof}
A large-dense vertex $v$ is adjacent to a vertex outside $V_{\mathrm{ld}}$ only when
\begin{itemize}
    \item[(i)] $v\in S$,
    \item[(ii)] $v\not \in S$ and $v$ has a neighbor in $S\setminus V_{\mathrm{ld}}$, or
    \item[(iii)] $v\not \in S$ and the component $C\in \mathcal{C}_{\mathrm{clique}}$ containing $v$ contains a vertex from $V\setminus V_{\mathrm{ld}}$.
\end{itemize}
Reduction Rule~\ref{rrule:boundS} ensures that at most $|S|\leq 4k$ vertices satisfy condition (i).
Moreover, since vertices in $V\setminus V_{\mathrm{ld}}$ has degree at most $7k$, at most $7k|S|\leq 28k^2$ vertices satisfy condition (ii).
Furthermore, if $v$ satisfies condition (iii), $C$ is a clique of size at most $7k+1$ because it contains a vertex with degree at most $7k$. Therefore, at most $7k|\mathcal{C}_{\mathrm{clique}}|\leq 56k^2$ vertices satisfy condition (iii).
\end{proof}

Now we bound the number of vertices $v\in V_{\mathrm{ld}}$ with $N_G(v)\subseteq V_{\mathrm{ld}}$.
A vertex triplet $(v_1,v_2,v_3)$ \emph{induces $P_3$} if $(v_1,v_2),(v_2,v_3)\in E$ and $(v_3,v_1)\not \in E$.
We remark that even if the edges $(v_1,v_2)$ or $(v_2,v_3)$ are multi-edges, we still consider them as inducing a $P_3$.
The following lemma is analogous to the fundamental observation in the literature on \textsc{Cluster Vertex Deletion}.
\begin{lemma}\label{lem:hitsp3}
Let $X$ be a feasible solution and $P$ be a subset of $V_{\mathrm{ld}}$ that induces $P_3$. Then, $X\cap P\neq \emptyset$.
\end{lemma}
\begin{proof}
Assume otherwise and let $C$ be the connected component of $G-X$ that contains $P$.
From Lemma~\ref{lem:dense_is_clique}, $C$ is a clique component. However, cliques cannot contain induced $P_3$, leading to a contradiction.
\end{proof}
This leads to the following reduction rule, which is clearly safe.
\begin{rrule}\label{rrule:manyp3}
Let $v\in V_{\mathrm{ld}}$. If there is a collection of $k+1$ induced $P_3$s in $V_{\mathrm{ld}}$ such that any two of them intersect only at $\{v\}$, then remove $v$ and decrease $k$ by $1$.
\end{rrule}

Whether $v\in V$ satisfies the condition of Reduction Rule~\ref{rrule:manyp3} can be checked by computing the maximum matching on the graph $(N_G(v),E_v)$, where $E_v$ is the set of pairs of the vertices $(u_1,u_2)$ such that $\{v,u_1,u_2\}$ induces $P_3$. Therefore, this rule can be applied in polynomial time.
Let $\mathcal{P}$ be a maximal collection of induced $P_3$s in $V_{\mathrm{ld}}$ such that any two induced $P_3$s in the collection have an intersection of size at most $1$.
The idea to construct this $\mathcal{P}$ and the following reduction rule are borrowed from the quadratic kernelization of \textsc{3-Hitting Set} by Abu-Khzam~\cite{abu2010kernelization}.

\begin{rrule}\label{rrule:largep3s}
If $|\mathcal{P}| > k^2$, return $\mathsf{no}$.
\end{rrule}
\begin{lemma}\label{lem:largep3s_is_safe}
Reduction Rule~\ref{rrule:largep3s} is safe.
\end{lemma}
\begin{proof}
Assume $|\mathcal{P}|>k^2$ and let $X$ be a feasible solution.
From the assumption that Reduction Rule~\ref{rrule:manyp3} cannot be applied, each vertex in $X$ hits at most $k$ induced $P_3$s in $\mathcal{P}$.
Therefore, there exists an induced $P_3$ that is disjoint from $X$, which contradicts Lemma~\ref{lem:hitsp3}.
\end{proof}

Now we can assume $\left|\bigcup_{P\in \mathcal{P}}P\right|\leq 3|\mathcal{P}|\leq 3k^2$.
The remaining task is to bound the number of vertices in $V_{\mathrm{ldmod}}:=V_{\mathrm{ld}}\setminus \bigcup_{P\in \mathcal{P}}P$.
We first bound the number of connected components.
\begin{lemma}\label{lem:module_components_bound}
$V_{\mathrm{ldmod}}$ contains at most $12k$ connected components. 
\end{lemma}
\begin{proof}
A clique can intersect at most one connected component of $V_{\mathrm{ldmod}}$.
Since Reduction Rule~\ref{rrule:expansion_reduction_clique} is exhaustively applied, $V_{\mathrm{ld}}$ can be partitioned into at most $8k$ cliques and $|S|\leq 4k$ vertices.
Therefore, there can be at most $8k+4k = 12k$ connected components in $V_{\mathrm{ldmod}}$.
\end{proof}

We state that each connected component of $V_{\mathrm{ldmod}}$ has a specific structure.
A vertex set $Z$ is \emph{extended clique-module} in a graph $G'=(V',E')$ if 
\begin{itemize}
    \item[(i)] for each different $u,v\in Z$, $(u,v)\in E'$, and
    \item[(ii)] for each $u,v\in Z$, $N_{G'}(u)\cup \{u\}=N_{G'}(v)\cup \{v\}$.
\end{itemize}
In other words, an extended clique-module is a \emph{clique-module}~\cite{fomin2019kernelization} in the graph obtained by reducing the multiplicity of each multi-edge to $1$.
We can prove that each connected component of $V_{\mathrm{ldmod}}$ is actually an extended clique-module in $V_{\mathrm{ld}}$.
\begin{lemma}\label{lem:extended_clique_module}
Each connected component of $V_{\mathrm{ldmod}}$ is an extended clique-module in $V_{\mathrm{ld}}$.
\end{lemma}
\begin{proof}
Let $C$ be a connected component of $V_{\mathrm{ldmod}}$.
$C$ should satisfy condition~(i), because otherwise $C$ would contain an induced $P_3$, contradicting the maximality of $\mathcal{P}$.
Moreover, $C$ should satisfy condition~(ii), because otherwise there would exist $u\in \bigcup_{P\in \mathcal{P}}P$ and $v_1,v_2\in C$ such that $(u,v_1)\in E$ and $(u,v_2)\not \in E$, which would form an induced $P_3$, again contradicting the maximality of $\mathcal{P}$.
\end{proof}

Let $E_{\mathrm{mul}}$ be the set of multi-edges of $G$.
Since any feasible solution should contain at least one endpoint of each edge in $E_{\mathrm{mul}}$, the graph $G_{\mathrm{mul}}:=(V,E_{\mathrm{mul}})$ should have a vertex cover of size $k$. This observation leads to the following two reduction rules called \emph{Buss rule}~\cite{buss1993nondeterminism,cygan2015parameterized,fomin2019kernelization} in the literature of kernelization of \textsc{Vertex Cover}, which are clearly safe.
\begin{rrule}\label{rrule:multiedge_vc_1}
If there is a vertex $v\in V$ with $|N_{G_{\mathrm{mul}}}(v)| > k$, remove $v$ and decrease $k$ by $1$.
\end{rrule}
\begin{rrule}\label{rrule:multiedge_vc_2}
If $|E_{\mathrm{mul}}| > k^2$, return $\mathsf{no}$.
\end{rrule}

We finish the analysis by bounding the number of vertices in each connected component of $V_{\mathrm{ldmod}}$ that is neither adjacent to a vertex outside $V_{\mathrm{ld}}$ nor incident to a multi-edge.
We have the following.

\begin{lemma}\label{lem:equivalent_vertices}
Let $u,v\in V_{\mathrm{ld}}$ with $(u,v)\in E$ and assume neither $u$ nor $v$ is incident to a multi-edge.
Assume $N_G(u)\cup \{u\} = N_G(v)\cup \{v\}$.
Then, for any minimum feasible solution $X$, either $\{u,v\}\subseteq X$ or $\{u,v\}\cap X=\emptyset$ holds.
\end{lemma}
\begin{proof}
Assume a feasible solution $X$ satisfies $u\in X$ and $v\not \in X$.
It is sufficient to prove that $X\setminus \{u\}$ is still feasible.
Let $C$ be a connected component of $G-(X\setminus \{u\})$ containing both $u$ and $v$.
Since $X$ is feasible, all connected components of $G-(X\setminus \{u\})$ other than $C$ are either cliques or trees.
It suffices to show that $C$ is a clique.
Since $v\in V_{\mathrm{ld}}$, the connected component $C'$ of $G-X$ that contains $v$ is a clique.
Since $N_{G-(X\setminus \{u\})}(u)=N_{G-(X\setminus \{u\})}(v)\cup \{v\}\setminus \{u\} = N_{G-X}(v)\cup \{v\}\setminus \{u\} = C'$, we obtain $C=C'\cup \{u\}$.
Since there are no multi-edges between $u$ and $C'$, $C$ is a clique.
\end{proof}

Now, we apply the following reduction rule.
\begin{rrule}\label{rrule:large_clique_module}
Assume a connected component of $V_{\mathrm{ldmod}}$ contains $k+4$ vertices that are not adjacent to vertices outside $V_{\mathrm{ld}}$ and incident to no multi-edges.
Then, remove one of those $k+4$ vertices.
\end{rrule}
\begin{lemma}\label{lem:large_clique_module_is_safe}
Reduction Rule~\ref{rrule:large_clique_module} is safe.
\end{lemma}
\begin{proof}
Let $Z$ be a set of vertices satisfying the assumption of the rule.
Clearly, $Z$ induces a clique.
Since each vertex in $Z$ has no neighbor outside $V_{\mathrm{ld}}$, from Lemma~\ref{lem:extended_clique_module}, the set $N_G(v)\cup \{v\}$ is same for all $v\in Z$.
Let $X$ be a minimum feasible solution for $G$.
Then, Lemma~\ref{lem:equivalent_vertices} ensures $Z\cap X=\emptyset$ because $|Z|>k$.
Therefore, $X$ is still a feasible solution for $G-v$ for any $v\in Z$.
Conversely, let $v\in Z$ and $X'$ be a feasible solution for $G-v$.
Since $Z\setminus \{v\}$ is a clique of size $k+3$, $Z-v-X'$ is contained in a single clique component $C'$ of $G-v-X'$.
Moreover, we have $N_{G-X'}(v)=N_{G-X'}(u)\cup \{u\}\setminus \{v\}$ for all $u\in Z\setminus (X'\cup \{v\})$, and thus, $N_{G-X'}(v)=C'$. Particularly, $C'\cup \{v\}$ is a clique component of $G-X'$.
Therefore, $X'$ is still a feasible solution for $G$.
\end{proof}

Now we can bound $|V_{\mathrm{ld}}|$ by $O(k^2)$.

\begin{lemma}\label{lem:bound_large_dense}
After applying all the above reduction rules, we have $|V_{\mathrm{ld}}|\leq 101k^2+40k$.
\end{lemma}
\begin{proof}
A vertex $v$ is in $V_{\mathrm{ld}}$ only when
\begin{itemize}
    \item[(i)] it has a neighbor outside $V_{\mathrm{ld}}$,
    \item[(ii)] it is in $\bigcup_{P\in \mathcal{P}}P$,
    \item[(iii)] there is a multi-edge incident to $v$, or
    \item[(iv)] it does not satisfy neither of condition~(i),~(ii),~and~(iii).
\end{itemize}
Lemma~\ref{lem:large_dense_outside} states that at most $84k^2+4k$ vertices satisfy condition (i).
Reduction Rule~\ref{rrule:largep3s} ensures that at most $3k^2$ vertices satisfy condition (ii).
Reduction Rule~\ref{rrule:multiedge_vc_2} ensures that at most $2k^2$ vertices satisfy condition (iii).
Lemma~\ref{lem:module_components_bound} and Reduction Rule~\ref{rrule:large_clique_module} ensures at most $12k\cdot (k+3) = 12k^2+36k$ vertices satisfy condition (iv).
Therefore, we have $|V_{\mathrm{ld}}|\leq (84k^2+4k) + 3k^2 + 2k^2 + (12k^2+36k)\leq 101k^2+40k$.
\end{proof}

\subsection{Bounding Number of Vertices}\label{sec:combine_all}

Now, we finish the overall analysis by bounding the number of vertices in $V\setminus V_{\mathrm{ld}}$.
We first bound the number of edges between $S$ and the tree components of $G-S$ (remember, $S$ is a $4$-approximate solution).
Let $V_{\mathrm{tree}}:=\bigcup_{C\in \mathcal{C}_{\mathrm{tree}}}C$.
We begin by the following.

\begin{rrule}\label{rrule:large_tree_neighbors}
If there is a vertex $v\in V_{\mathrm{ld}}\cap S$ such that $N(v)$ contains at least $2k+4$ vertices from $V_{\mathrm{tree}}$, remove $v$ and decrease $k$ by $1$.
\end{rrule}
\begin{lemma}\label{lem:large_tree_is_safe}
Reduction Rule~\ref{rrule:large_tree_neighbors} is safe.
\end{lemma}
\begin{appendixproof}
We prove that any feasible solution should contain $v$.
Since $v\in V_{\mathrm{ld}}$, $v$ cannot be contained in a tree component of $G-X$.
Since every forest contains an independent set consisting of at least half of its vertices, $N(v)$ contains an independent set of size $k+2$.
Therefore, $v$ cannot be contained in a clique component of $G-X$.
\end{appendixproof}

Thus, we can bound the number of vertices in $V_{\mathrm{tree}}$ adjacent to some vertex in $S$.

\begin{lemma}\label{lem:degree_to_trees}
After exhaustively applying all the above reduction rules, at most $28k^2$ vertices in $V_{\mathrm{tree}}$ are adjacent to some vertex in $S$.
\end{lemma}
\begin{appendixproof}
From Lemma~\ref{lem:degree_bound_trees} and Reduction Rule~\ref{rrule:large_tree_neighbors}, we have that for each $v\in S$, $\left|N(v)\cap V_{\mathrm{tree}}\right|\leq 7k$. Therefore, at most $7k\cdot |S|\leq 28k^2$ vertices in $V_{\mathrm{tree}}$ are adjacent to some vertex in $S$.
\end{appendixproof}

We apply the following reduction rules by local structures.
Most of them appear in~\cite{jacob2024kernel}, but Reduction Rule~\ref{rrule:pendant_removal} is new.
\begin{rrule}[\rm{\cite{jacob2024kernel}}]\label{rrule:two_pendants_removal}
If a vertex $v$ is adjacent to at least two pendant vertices, remove one of them.
\end{rrule}

\begin{rrule}[\rm{\cite{jacob2024kernel}}]\label{rrule:multiedge_removal}
If there are multi-edges with a multiplicity of at least $3$, reduce the multiplicity to $2$.
\end{rrule}

\begin{rrule}[\rm{\cite{jacob2024kernel}}]\label{rrule:tail_removal}
If there are three different vertices $(v_1,v_2,v_3)$ such that $(v_1,v_2), (v_2,v_3)\in E$, $(v_1,v_3)\not \in E$, $d(v_2)=2$, and $d(v_3)=1$, remove $v_3$.
\end{rrule}

\begin{rrule}[\rm{\cite{jacob2024kernel}}]\label{rrule:p5_removal}
If there are five different vertices $(v_1,v_2,v_3,v_4,v_5)$ such that $(v_1,v_2)$, $(v_2,v_3)$, $(v_3,v_4)$, $(v_4,v_5)$ are in $E$ and $d(v_2)=d(v_3)=d(v_4)=2$, remove $v_3$ and add an edge $(v_2,v_4)$.
\end{rrule}

\begin{rrule}\label{rrule:pendant_removal}
If there are four different vertices $(v_1,v_2,v_3,v_4)$ such that $(v_i,v_j)\in E$ if and only if $i=1$, $d(v_1)=3$, and $d(v_4)=1$, remove $v_4$.
\end{rrule}
\begin{lemma}\label{lem:pendant_removal_is_safe}
Reduction Rule~\ref{rrule:pendant_removal} is safe.
\end{lemma}
\begin{appendixproof}
Let $G=(V,E)$ be a graph and $(v_1,v_2,v_3,v_4)$ be vertices taken in the rule.
Clearly, if $G$ is an $\mathsf{yes}$-instance, $G-v_4$ is also an $\mathsf{yes}$-instance.
We prove the opposite direction.
Let $X'$ be a feasible solution for $G-v_4$.
If $v_1\in X'$, $X'$ is also a feasible solution for $G$ because $v_4$ is isolated in $G-X'$.
Assume $v_1\not \in X'$. It is sufficient to prove that the component $C$ containing $v_1$ in $G-v_4-X'$ is a tree component because if it holds, we can safely add a pendant vertex $v_4$ to this component, and particularly, $X'$ is a feasible solution for $G$.
If at least one of $v_2$ or $v_3$ is in $X'$, we have $d_{G-v_4-X'}(v_1)\leq 1$ and $C$ is a tree component. Otherwise, both $v_2$ and $v_3$ are in $C$ but we have $(v_2,v_3)\not \in E$, and thus, $C$ is a tree component.
\end{appendixproof}






Now, we bound $|V_{\mathrm{tree}}|$. We begin with the following.

\begin{lemma}\label{lem:pendant_three}
After applying all the above reduction rules, $V_{\mathrm{tree}}$ contains at most $28k^2$ pendant vertices that are adjacent to a vertex of degree $3$.
\end{lemma}
\begin{appendixproof}
Let $v\in V_{\mathrm{tree}}$ be a pendant vertex and $u$ be the vertex adjacent to it.
Since Reduction Rules~\ref{rrule:two_pendants_removal}~and~\ref{rrule:pendant_removal} can no longer be applied, $u$ is contained in a cycle $P$ of length $3$.
Since $v$ is in a tree component of $G-S$, we have $P\cap S\neq \emptyset$.
Therefore, $u$ is either in $S$ or adjacent to some vertex in $S$.
Let $w:=v$ if $u\in S$ and $w:=u$ otherwise. Then, $w$ is contained in a tree component and adjacent to some vertex in $S$.
From Lemma~\ref{lem:degree_to_trees}, there are at most $28k^2$ candidates of $w$. From Reduction Rule~\ref{rrule:two_pendants_removal}, each candidate of $w$ corresponds to at most one candidate of $v$. Therefore, the lemma is proved.
\end{appendixproof}

Now, we can prove the following.
\begin{lemma}\label{lem:tree_vertices}
We have $|V_{\mathrm{tree}}|\leq 1232k^2$.
\end{lemma}
\begin{appendixproof}
Let $V'_{\mathrm{tree}}$ be the set of vertices in $V_{\mathrm{tree}}$ that are adjacent to some vertex in $S$.
Let $V_1$, $V_2$, $V_3$, and $V_{\geq 4}$ be the set of vertices in $V_{\mathrm{tree}}\setminus V'_{\mathrm{tree}}$ with degree exactly $1$, exactly $2$, exactly $3$, and at least $4$, respectively.
Since $G$ does not contain isolated vertices because of Reduction Rule~\ref{rrule:trivial_component_removal}, $V_1\dot{\cup }V_2\dot{\cup }V_3\dot{\cup }V_{\geq 4}\dot{\cup }V'_{\mathrm{tree}}$ is a partition of $V_{\mathrm{tree}}$.

From Lemma~\ref{lem:degree_to_trees}, we have $|V'_{\mathrm{tree}}|\leq 28k^2$.
From Lemma~\ref{lem:pendant_three}, we have $|V_1|\leq (|V_{\geq 4}|+|V'_{\mathrm{trees}}|)+28k^2\leq |V_{\geq 4}|+56k^2$.
Since the subgraph induced by $V_{\mathrm{tree}}$ contains at most $2|V_{\mathrm{tree}}|$ edges, we have
\begin{align*}
    |V_1|+2|V_2|+3|V_3|+4|V_{\geq 4}|\leq 2|V_{\mathrm{tree}}|\leq 2(|V_1|+|V_2|+|V_3|+|V_{\geq 4}|)+56k^2,
\end{align*}
which leads $|V_3|+2|V_{\geq 4}|\leq |V_1|+56k^2$.
Particularly, we have 
\begin{align*}
    2|V_{\geq 4}|\leq |V_3|+2|V_{\geq 4}|\leq |V_1|+56k^2\leq |V_{\geq 4}|+112k^2,
\end{align*}
which results in $|V_{\geq 4}|\leq 112k^2$ and $|V_1|\leq 168k^2$.
Therefore, we have
\begin{align*}
    |V_{\mathrm{tree}}|&=|V_1|+|V_2|+|V_3|+|V_{\geq 4}|+|V'_{\mathrm{tree}}|\\
    &\leq |V_2|+|V_3|+2|V_{\geq 4}|+84k^2
    \leq |V_2|+|V_1|+140k^2\leq |V_2|+308k^2,
\end{align*}
where the first inequality is from $|V_1|\leq |V_4|+56k^2$ and $|V'_{\mathrm{tree}}|\leq 28k^2$.

Now, we complete the proof by bounding $|V_2|$.
From Reduction Rule~\ref{rrule:p5_removal}, each connected component of $V_2$ has size at most $3$.
Let $(V_T, E_T)$ be the forest obtained from $V_{\mathrm{tree}}$ by contracting all vertices in $V_2$. Then, we have
\begin{align*}
    |V_2|\leq 3|E_T|\leq 3|V_T| = 3(|V_{\mathrm{tree}}|-|V_2|)\leq 924k^2.
\end{align*}
That yields $|V_{\mathrm{tree}}|\leq 308k^2+924k^2=1232k^2$.
\end{appendixproof}

Now, we can bound the size of the whole graph, which directly proves Theorem~\ref{thm:main}.
\begin{lemma}\label{lem:overall_analysis}
After applying all the above reduction rules, we have $|V|\leq 1389k^2+52k$.
\end{lemma}
\begin{appendixproof}
Each vertex $v\in V$ satisfies either
\begin{itemize}
    \item[(i)] $v\in V_{\mathrm{ld}}$,
    \item[(ii)] $v\not \in V_{\mathrm{ld}}$ and $v\in S$,
    \item[(iii)] $v\not \in V_{\mathrm{ld}}$ and $v$ is contained in a clique component of $G-S$, or
    \item[(iv)] $v\not \in V_{\mathrm{ld}}$ and $v$ is contained in a tree component of $G-S$.
\end{itemize}
Lemma~\ref{lem:bound_large_dense} ensures that at most $101k^2+40k$ vertices satisfy condition (i).
Clearly, at most $|S|\leq 4k$ vertices satisfy condition (ii).
Moreover, since all vertices not in $V_{\mathrm{ld}}$ has a degree at most $7k$ and there are at most $8k$ clique components in $G-S$, at most $(7k+1)\cdot 8k\leq 56k^2+8k$ vertices satisfy condition (iii).
Furthermore, Lemma~\ref{lem:tree_vertices} ensures that at most $1232k^2$ vertices satisfy condition (iv).
Therefore, we have $|V|\leq (101k^2+40k)+4k+(56k^2+8k)+1232k^2=1389k^2+52k$.
\end{appendixproof}

\section{FPT Algorithm for Longest Cycle Parameterized by Cliques or Trees Vertex Deletion Number}\label{sec:cycle}

In this section, we prove that \textsc{Longest Cycle} is fixed-parameter tractable when parameterized by the cliques or trees vertex deletion number.

\subsection{Overall Strategy}\label{sec:cycle_strategy}
Let $G=(V,E)$ be a graph, let $k\in \mathbb{R}_{\geq 0}$ be an integer, and let $S\subseteq V$ be a set with $|S| = k$ such that each connected component of $G-S$ is either a clique or a tree.
We assume that $G$ is not a forest; otherwise, there is no cycle.
Let $\mathcal{C}_{\mathrm{clique}}$ and $\mathcal{C}_{\mathrm{tree}}$ be the set of clique components and tree components, respectively.
For a cycle $P$ in $G$, a \emph{fragment} of $P$ is a subpath of $P$ whose endpoints belong to $S$ and do not contain any other vertex from $S$.
If $P$ is not disjoint from $S$, then $P$ can be decomposed into fragments that share no internal vertices.

The maximum length of a cycle disjoint from $S$ is equal to the maximum size of the clique components of $G-S$, which is clearly computable in polynomial time.
Hereafter, we focus on finding the longest cycle that is not disjoint from $S$.
Then, each fragment of $P$ satisfies that, unless its length is $1$, its internal vertex set belongs to a single component in $\mathcal{C}_{\mathrm{clique}}\cup \mathcal{C}_{\mathrm{tree}}$.

For (possibly same) $u, v\in S$ and $C\in \mathcal{C}_{\mathrm{clique}}\cup \mathcal{C}_{\mathrm{tree}}$, we define the value $q(C,u,v)$ as follows:
\begin{itemize}
    \item[(i)] If $N_G(u)\cap C=\emptyset$ or $N_G(v)\cap C=\emptyset$, we define $q(C,u,v)=-\infty$.
    \item[(ii)] If $C\in \mathcal{C}_{\mathrm{clique}}$ and $N_G(u)\cap C=N_G(v)\cap C=\{w\}$ for some $w\in C$, then we define $q(C,u,v)=2$.
    \item[(iii)] If $C\in \mathcal{C}_{\mathrm{clique}}$ and conditions~(i) and~(ii) do not hold, that is, if there exist distinct vertices $w_u, w_v\in C$ such that $(u,w_u),(v,w_v)\in E$, we define $q(C,u,v)=|C|+1$.
    \item[(iv)] If $C\in \mathcal{C}_{\mathrm{tree}}$ and condition~(i) does not hold, let $w_u, w_v\in C$ be the pair of (possibly same) vertices such that $(u,w_u), (v,w_v)\in E$ and maximizes the distance between $w_u$ and $w_v$ along the tree $C$. Define $q(C,u,v)$ as the length of the $w_u,w_v$-path plus two.
\end{itemize}
In other words, $q(C,u,v)$ is the maximum length of a $u,v$-path passing through the vertices of $C$, if such a path exists.
For each (possibly same) $u,v\in S$, we select the top $k$ components with the highest positive $q(C,u,v)$ values among all $C\in \mathcal{C}_{\mathrm{clique}}\cup \mathcal{C}_{\mathrm{tree}}$, arbitrarily breaking ties if necessary, and add label $(u,v)$ to these components. 
If fewer than $k$ components have positive $q(C,u,v)$ values, we label all such components.
We have the following lemma.

\begin{lemma}\label{lem:labeled_components}
Let $P$ be a cycle of $G$ that is not disjoint from $S$.
Let its fragments be $P_1,\dots, P_l$ in order appearing in $P$, and for each $i\in \{1,\dots, l\}$, let $v_i\in S$ be the vertex so that $P_i$ is a $v_i-v_{i+1}$ path, where indices are considered in modulo $l$.
Then, there exists a cycle $P'$ in $G$ that can be decomposed into $l$ fragments $P'_1,\dots, P'_l$ such that
\begin{itemize}
    \item $P'_i$ is a $v_i-v_{i+1}$ path,
    \item internal vertices of $P'_i$ belong to a component labeled $(v_i,v_{i+1})$ (if exist),
    \item $|P|\leq |P'|$.
\end{itemize}
\end{lemma}

\begin{proof}
Let $i$ be an index such that $P_i$ has internal vertices.
Then, we have $q(C_i,v_i,v_{i+1}) > 0$. 
If $C_i$ is not labeled with $(v_{i},v_{i+1})$, then at least $k$ components $C\in \mathcal{C}_{\mathrm{clique}}\cup \mathcal{C}_{\mathrm{tree}}$ are labeled with $(v_{i},v_{i+1})$.
Therefore, there exists a component $C'_i$ disjoint from $P$, labeled with $(v_i,v_{i+1})$, such that $q(C'_i,v_i,v_{i+1})\geq q(C_i,v_i,v_{i+1})$.
We replace $P_i$ with a $v_{i}-v_{i+1}$ path of length $q(C'_i,v_i,v_{i+1})+1$.
This replacement does not decrease the length of $P$, and increases the number of fragments labeled $(v_i,v_{i+1})$. Thus, iterating this process exhaustively yields the desired cycle $P'$.
\end{proof}

Our algorithm first guesses the structure of the longest cycle $P$.
Specifically, we guess the sequence $(v_1,\dots, v_l)$ of distinct vertices in $S$.
For each $i\in \{1,\dots, l\}$, we also guess a component $C_i\in \mathcal{C}_{\mathrm{clique}}\cup \mathcal{C}_{\mathrm{tree}}$ labeled with $(v_i,v_{i+1})$, which represents that internal vertices of $P_i$ are in $C_i$, or set $C_i=\bot$, which represents that $P_i$ has length one.
Since there are at most $k\cdot k!$ candidate sequences for $(v_1,\dots, v_{l})$ and $k^k$ choices for $(C_1,\dots, C_l)$, this guessing step costs $2^{O(k\log k)}$ to the time complexity.
We define the following \textsc{Longest Disjoint Terminal Paths}.

\begin{quote}
\textsc{Longest Disjoint Terminal Paths}:
Given a graph that is either a clique or a tree, and given a list of vertex subset pairs $((V_{11},V_{12}),\dots, (V_{l1},V_{l2}))$, compute the maximum total length of $l$ vertex-disjoint paths $(P_1,\dots, P_l)$ such that each $P_i$ is a path from a vertex in $V_{i1}$ to a vertex in $V_{i2}$.
\end{quote}

Here we see how to reduce \textsc{Longest Cycle} to repeatedly solving \textsc{Longest Disjoint Terminal Paths}.
First, if there is $i\in \{1,\dots, l\}$ such that $C_i=\bot$ and $(v_i,v_{i+1})\not \in E$, no solution exists for this guess.
We also treat the corner case where $l = 2$, $C_1 = C_2 = \bot$, and $(v_1, v_2)$ is a single edge in this part by returning that no solution exists.  
Assume otherwise. 
For each $C\in \mathcal{C}_{\mathrm{clique}}\cup \mathcal{C}_{\mathrm{tree}}$, let $J_C$ be the set of indices $i$ such that $C_i=C$.
For each $C$ with $J_C\neq \emptyset$, we call \textsc{Longest Disjoint Terminal Paths} for the list of vertex subset pairs $\{(N(v_i)\cap C, N(v_{i+1})\cap C)\}_{i\in J_C}$.
Then, the length of the longest cycle for this guess is given by the sum of the returned values of these calls plus $l+|\{i\in \{1,\dots, l\}\colon C_i\neq \bot\}|$.

In the following sections, we present FPT algorithms for \textsc{Longest Disjoint Terminal Paths} parameterized by $l$, separately for cliques and trees, thereby completing the proof of Theorem~\ref{thm:cycle}.

\subsection{Longest Disjoint Terminal Paths on Cliques}
Here, we solve \textsc{Longest Disjoint Terminal Paths} on cliques.
Let $V$ be a vertex set, $l\in \mathbb{Z}_{\geq 1}$, and $((V_{11},V_{12}),\dots, (V_{l1},V_{l2}))$ be a list of pairs of subsets of $V$.
Our algorithm guesses an array of flags $f:=(f_1,\dots, f_l)\in \{0,1\}^l$, where $f_i=0$ indicates that $P_i$ consists of a single vertex, and $f_i=1$ indicates that $P_i$ contains at least one edge.
Obviously, this guess costs $O(2^l)$ to the time complexity.
We define the multifamily $\mathcal{D}_f$ as follows.
\begin{align*}
    \mathcal{D}_f:=\bigcup_{\substack{i\in \{1,\dots, l\},\\f_i=0}}\{V_{i1}\cap V_{i2}\}\cup \bigcup_{\substack{i\in \{1,\dots, l\},\\f_i=1}}\{V_{i1}, V_{i2}\}.
\end{align*}
We further define the bipartite graph $H_f$ with bipartition $(\mathcal{D}_f,V)$, where an edge $(D,v)$ exists if and only if $v\in D$.
We have the following.
\begin{lemma}
$H_f$ contains a matching of size $|\mathcal{D}_f|$ if and only if there exists a list of disjoint paths $(P_1,\dots, P_l)$ that satisfies, for each $i\in \{1,\dots, l\}$,
\begin{itemize}
    \item if $f_i=0$, then $P_i$ consists of a single vertex; if $f_i=1$, then $P_i$ contains at least one edge, and
    \item one endpoint of $P_i$ belongs to $V_{i1}$ and the other endpoint belongs to $V_{i2}$.
\end{itemize}
\end{lemma}
\begin{proof}
Let $M$ be a matching of size $|\mathcal{D}_f|$ in $H_f$. 
For each $D\in \mathcal{D}_f$, let $\sigma(D)$ be the vertex matched to $D$ in $M$.
Define $P_i$ as a path consisting of the single vertex $\sigma(V_{i1}\cap V_{i2})$ if $f_i=0$, and as a path consisting of the single edge $(\sigma(V_{i1}),\sigma(V_{i2}))$ if $f_i=1$.
Then, $(P_1,\dots, P_l)$ satisfies the conditions of the lemma.
Conversely, suppose there exists a list of disjoint paths $(P_1,\dots, P_l)$ satisfying the conditions.
Since these paths are disjoint, their endpoint sets are also disjoint.
Thus, defining $\sigma(V_{i1}\cap V_{i2})$ as the only vertex of $P_i$ if $f_i=0$ and $\sigma(V_{i1})$ and $\sigma(V_{i2})$ as endpoints of $P_i$ if $f_i=1$, we obtain a matching of size $|\mathcal{D}_f|$ in $H_f$.
\end{proof}

Our algorithm constructs $H_f$ and checks whether it contains a matching of size $|\mathcal{D}_f|$.
If no such matching exists in $H_f$ for all guesses $f:=(f_1,\dots, f_l)$, the algorithm reports that \textsc{Longest Disjoint Terminal Paths} has no solution.
If there exists some $f\neq (0,\dots,0)$ such that $H_f$ contains a matching of size $|\mathcal{D}_f|$, the algorithm returns $|V|-l$; that is, the algorithm reports that there is a list of paths that covers all vertices in $V$.
The correctness is from the following observation: If there exists a vertex not included in any of $P_1,\dots, P_l$, we can extend a path $P_i$ with $f_i=1$ by adding this vertex as an internal vertex.
For the remaining case, that is, if a matching of size $|\mathcal{D}_f|$ exists only for $f=(0,\dots,0)$, the algorithm returns $0$.
By this algorithm, we have the following.
\begin{lemma}\label{lem:disjoint_clique}
\textsc{Longest Disjoint Terminal Paths} admits an $O^*\left(2^l\right)$-time algorithm on cliques.
\end{lemma}

\subsection{Longest Disjoint Terminal Paths on Trees}
Here, we solve \textsc{Longest Disjoint Terminal Paths} on trees using dynamic programming.
Let $T:=(V,E)$ be a tree, $l\in \mathbb{Z}_{\geq 1}$, and $((V_{11},V_{12}),\dots, (V_{l1},V_{l2}))$ be a list of pairs of subsets of $V$.
We select an arbitrary vertex $r$ as the root and consider $T$ as a rooted tree.
For a vertex $v\in V$, let $T_v$ denote the subtree rooted at $v$.
For $j\in \{1,\dots, l\}$, a path $P_j$ in $T$ is \emph{proper} if one endpoint belongs to $V_{j1}$ and the other belongs to $V_{j2}$.
For $v\in V$ and $f\in \{1,2\}$, a path $P_j$ is \emph{$(v,f)$-semiproper} if one endpoint belongs to $V_{jf}$ and the other endpoint is $v$.

For $v\in V$, $Z\subseteq \{1,\dots, l\}$, $i\in \{1,\dots, l\}\setminus Z$, and $f\in \{1,2\}$, we define $\DP_1[v][Z][i][f]$ as the maximum possible total length among all lists of paths $\{P_j\}_{j\in Z}\cup \{P_i\}$ that use only vertices in $T_v$, such that each path in $\{P_j\}_{j\in Z}$ is proper and $P_i$ is $(v,f)$-semiproper.
Similarly, for $v\in V$ and $Z\subseteq \{1,\dots, l\}$, we define $\DP_2[v][Z]$ as the maximum possible total length among all lists of proper paths $\{P_j\}_{j\in Z}$ that use only vertices in $T_v$.
The answer to the \textsc{Longest Disjoint Terminal Paths} is given by $\DP_2[r][\{1,\dots, l\}]$.

Below, we describe how to compute these DP arrays in a bottom-up manner.
Let $v\in V$ and $\{c_1,\dots, c_{t_v}\}$ be the set of children of $v$.
To compute $\DP_1[v][\cdot][\cdot][\cdot]$ and $\DP_2[v][\cdot]$, we introduce the following auxiliary DP arrays.
For each $s\in \{0,\dots, t_v\}$ and $Z\subseteq \{1,\dots, l\}$, we define $\AUX_{0}[s][Z]$ and $\AUX_{2}[s][Z]$ as the maximum possible total lengths among all lists of proper paths $\{P_j\}_{j\in Z}$ that use only vertices in $T_{c_1}\cup \dots \cup T_{c_s}$ and $T_{c_1}\cup \dots \cup T_{c_s} \cup \{v\}$, respectively.
Moreover, for $i\in \{1,\dots, l\}\setminus Z$ and $f\in \{1,2\}$, we define $\AUX_{1}[s][Z][i][f]$ as the maximum possible total length among all lists of paths $\{P_j\}_{j\in Z}\cup\{P_i\}$ that use only vertices in $T_{c_1}\cup \dots \cup T_{c_s}\cup \{v\}$, where each $P_j$ is proper and $P_i$ is $(v,f)$-semiproper.
All $\DP$ and $\AUX$ values are set to $-\infty$ if no valid path configurations exist.

We initialize the $\AUX$ values as follows.
\begin{align*}
    \AUX_0[0][Z]&:=\begin{cases}
        0 & (Z=\emptyset)\\
        -\infty & (Z\neq \emptyset)
    \end{cases},\\
    \AUX_1[0][Z][i][f]&:=\begin{cases}
        0 & (Z=\emptyset \land v\in V_{if})\\
        -\infty & (\text{otherwise})
    \end{cases},\\
    \AUX_2[0][Z]&:=\begin{cases}
        0 & (Z=\emptyset\lor Z=\{i\} \text{ for some } i\in \{1,\dots, l\} \text{ with } v\in V_{i1}\cap V_{i2})\\
        -\infty & (\text{otherwise})
    \end{cases}.
\end{align*}
Now, for $s\in \{1,\dots, t_v\}$, we update the DP values using the following formulas.
\begin{align}
    \AUX_0[s][Z]&=\max_{Z'\subseteq Z}\left(\AUX_0[s-1][Z']+\DP_2[c_s][Z\setminus Z']\right),\label{eq:update_path_1}\\
    \AUX_1[s][Z][i][f]&=\max\left(\begin{array}{l}
        \displaystyle\max_{Z'\subseteq Z}\left(\AUX_1[s-1][Z'][i][f]+\DP_2[c_s][Z\setminus Z']\right),\\
        \displaystyle\max_{Z'\subseteq Z}\left(\AUX_0[s-1][Z']+\DP_1[c_s][Z\setminus Z'][i][f]+1\right)
    \end{array}\right),\label{eq:update_path_2}\\
    \AUX_2[s][Z]&=\max\left(\begin{array}{l}
        \displaystyle\max_{Z'\subseteq Z}\left(\AUX_2[s-1][Z']+\DP_2[c_s][Z\setminus Z']\right),\\
        \displaystyle\max_{\substack{i\in Z,f\in \{1,2\},\\Z'\subseteq Z\setminus \{i\}}}  
        \left(\AUX_1[s-1][Z'][i][f]+\DP_1[c_s][Z\setminus (Z'\cup\{i\})][i][3-f]+1\right)
    \end{array}\right)\label{eq:update_path_3}.
\end{align}
Furthermore, by definition, for each $v\in V$, the final DP values are obtained as follows.
\begin{align*}
    \DP_1[v][Z][i][f] &= \AUX_{1}[t_v][Z][i][f],\\
    \DP_2[v][Z] &= \AUX_2[t_v][Z].
\end{align*}
We prove the correctness of these update formulas.
\begin{lemma}
Let $v\in V$ and assume the values of $\DP_1[c_s][Z][i][f]$ and $\DP_2[c_s][Z]$ are correctly computed for all children $c_s$ of $v$, $Z\subseteq \{1,\dots, l\}$, $i\in \{1,\dots, l\}\setminus Z$, and $f\in \{1,2\}$.
Then, above update formulas correctly computes $\DP_1[v][Z][i][f]$ and $\DP_2[v][Z]$ for all $Z\subseteq \{1,\dots, l\}$, $i\in \{1,\dots, l\}\setminus Z$, and $f\in \{1,2\}$.
\end{lemma}
\begin{proof}
Correctness of initializing $\AUX_{0}[0][Z]$, $\AUX_{1}[0][Z][i][f]$, and $\AUX_{2}[0][Z]$ are clear.
Now, we assume $\AUX_{0}[s-1][Z]$, $\AUX_{1}[s-1][Z][i][f]$, and $\AUX_{2}[s-1][Z]$ are correctly computed. We prove the update formulas~\eqref{eq:update_path_1},~\eqref{eq:update_path_2},~and~\eqref{eq:update_path_3} correctly compute $\AUX_{0}[s][Z]$, $\AUX_{1}[s][Z][i][f]$, and $\AUX_{2}[s][Z]$, respectively.

First, we consider $\AUX_{0}[s][Z]$.
Let $\{P_j\}_{j\in Z}$ be a list of proper paths that use only vertices in $T_{c_1}\cup \dots\cup T_{c_{s}}$.
Then, each path either contained in $T_{c_1}\cup\dots\cup  T_{c_{s-1}}$ or $T_{c_s}$. 
Let $Z'$ be the list of indices of such paths contained in $T_{c_1}\cup\dots\cup  T_{c_{s-1}}$.
From induction hypothesis, the total length of paths in $\{P_j\}_{j\in Z'}$ and $\{P_j\}_{j\in Z\setminus Z'}$ are at most $\AUX_{0}[s-1][Z']$ and $\DP_{2}[c_s][Z\setminus Z']$, respectively. Thus, the total length of paths in $\{P_j\}_{j\in Z}$ is at most $\AUX_{0}[s][Z]$.
Conversely, let $Z'$ be the set that attains maximum of~\eqref{eq:update_path_1}. Then, merging lists of paths that attains $\AUX_{0}[s-1][Z']$ and $\DP_{2}[c_s][Z\setminus Z']$ yields the list of proper paths that only use vertices in $T_{c_1}\cup \dots\cup T_{c_s}$ with the total length $\AUX_{0}[s][Z]$.

Next, we consider $\AUX_{1}[s][Z][i][f]$. 
Let $\{P_j\}_{j\in Z}\cup {P_i}$ be a list of paths that use only vertices in $T_{c_1}\cup \dots\cup T_{c_{s}}\cup \{v\}$, where each $P_j$ is proper and $P_i$ is $(v,f)$-semiproper.
Let $Z'$ be the list of indices of proper paths contained in $T_{c_1}\cup \dots\cup T_{c_{s-1}}\cup \{v\}$.
Remark that $P_i$ is contained in either $T_{c_1}\cup \dots\cup T_{c_{s-1}}\cup \{v\}$ or $T_{c_s}\cup \{v\}$.
For the former case, from induction hypothesis, the total length of paths in $\{P_j\}_{j\in Z'}\cup \{P_i\}$ and $\{P_j\}_{j\in Z\setminus Z'}$ are at most $\AUX_{1}[s-1][Z'][i][f]$ and $\DP_2[c_s][Z\setminus Z']$, respectively.
For the latter case, from induction hypothesis, the total length of paths in $\{P_j\}_{j\in Z'}$ and $\{P_j\}_{j\in Z\setminus Z'}\cup \{P_i\}$ are at most $\AUX_0[s-1][Z']$ and $\DP_1[c_s][Z\setminus Z'][i][f]+1$, respectively.
Thus, the total length of paths in $\{P_j\}_{j\in Z}$ is at most $\AUX_{1}[s][Z][i][f]$.
The reverse direction is proved in similar way as~\eqref{eq:update_path_1}.

Finally, we consider $\AUX_{2}[s][Z]$.
Let $\{P_j\}_{j\in Z}$ be a list of proper paths that use only vertices in $T_{c_1}\cup \dots\cup T_{c_{s}}\cup \{v\}$.
Then, at most one of these paths uses vertices from both $T_{c_1}\cup \dots\cup T_{c_{s-1}}\cup \{v\}$ and $T_{c_s}$.
If such path does not exist, let $Z'$ be the list of indices of proper paths contained in $T_{c_1}\cup \dots\cup T_{c_{s-1}}\cup \{v\}$. Then, from induction hypothesis, the total length of paths in $\{P_j\}_{j\in Z'}$ and $\{P_j\}_{j\in Z\setminus Z'}$ are at most $\AUX_{2}[s-1][Z']$ and $\DP_2[c_s][Z\setminus Z']$, respectively.
Thus, the total length of paths in $\{P_j\}_{j\in Z}$ is at most $\AUX_2[s][Z]$.
Assume otherwise. Let $i$ be the index of the path that uses vertices from both $T_{c_1}\cup \dots\cup T_{c_{s-1}}\cup \{v\}$ and $T_{c_s}$, and $u_1$ and $u_2$ be its endpoints, where $u_1\in T_{c_1}\cup \dots\cup T_{c_{s-1}}\cup \{v\}$ and $u_2\in T_{c_s}$.
Let $Q_1$ be the $u_1,v$-path and $Q_2$ be the $c_s,u_2$-path.
If $Q_1$ is $(v,1)$-semiproper and $Q_2$ is $(c_s,2)$-semiproper, we set $f:=1$. Otherwise, that is, if $Q_1$ is $(v,2)$-semiproper and $Q_2$ is $(v,1)$-semiproper, we set $f:=2$.
Let $Z'$ be the list of indices of proper paths contained in $T_{c_1}\cup \dots\cup T_{c_{s-1}}\cup \{v\}$.
From induction hypothesis, the total length of paths in $\{P_j\}_{j\in Z}\cup \{Q_1\}$ and $\{P_j\}_{j\in (Z\setminus (Z'\cup\{i\})}\cup \{Q_2\}$ are at most $\AUX_{1}[s-1][Z'][i][f]$ and $\DP_1[c_s][Z\setminus (Z'\cup\{i\})][i][3-f]$, respectively.
Thus, the total length of paths in $\{P_j\}_{j\in Z}$ is at most $\AUX_2[s][Z]$. The reverse direction is proved in similar way as~\eqref{eq:update_path_1}.
\end{proof}

Since the updates can be computed in $O^*\left(3^l\right)$ time, we obtain the following result.
\begin{lemma}\label{lem:disjoint_tree}
\textsc{Longest Disjoint Terminal Paths} admits an $O^*\left(3^l\right)$-time algorithm on trees.
\end{lemma}
Theorem~\ref{thm:cycle} follows from the discussion in Section~\ref{sec:cycle_strategy} combined with Lemmas~\ref{lem:disjoint_clique} and~\ref{lem:disjoint_tree}.

\bibliography{bib}
\end{document}